\documentclass[12pt]{article}%
\usepackage{amsmath}%
\usepackage{amsfonts}%
\usepackage{amssymb}%
\usepackage{graphicx}
\providecommand{\U}[1]{\protect\rule{.1in}{.1in}}
\newtheorem{theorem}{Theorem}

\newenvironment{proof}[1][Proof]{\noindent\textbf{#1.} }{\ \rule{0.5em}{0.5em}}
\begin{document}

\title{Competitive Search with a Faulty Satnav  (GPS): When Probability Matching is Rational}
\author{Steve Alpern$^{1}$ and Mark Broom$^{2}$ \\
$^{1}$ Warwick Business School, Warwick University, \\
Coventry CV4 7AL, UK \\
steve.alpern@wbs.ac.uk \\
$^{2}$ Department of Mathematics, City St.George's, \\
University of London EC1V 0HB, UK \\
Mark.Broom@city.ac.uk}
\maketitle

\begin{abstract}
A divisible treasure is located at a node $H$ of a network. From a given
start node a group of $n$ Searchers each seek to reach $H$ first, dividing
the treasure equally with the other first arrivers. This type of search game is called competitive search, where the conflict is not between hider and searcher but between searchers. Examples are search for oil deposits or for a pilot downed over enemy territory. In our model, the Searchers have a common Satnav (GPS) which points to $H$ at each branch node with a known probability $p<1$
and each Searcher chooses the probability $q$ with which they follow the
pointer. We consider a family of star graphs where the Searchers start at the
center and $H$ lies at one of the $k$ leaf nodes. We show that for all
parameter values $n,k,p,$ there is a unique trust probability $q$ which forms
a symmetric equilibrium. The equilibrium $q$ is increasing in $p,$ decreasing
in $n$ and increasing in $k$. Furthemore for fixed $k$ and $p$ we have $q$ equal to $p$
in the limit of $n$ tending to infinity. This last fact is a new example where what is known
in behavioural science as \textit{probability matching} is in fact rational.

\end{abstract}

\section{Introduction}

We consider a network search game where a group of $n$ Searchers each tries to
be the \textit{first} to reach a divisible treasure at a node $H,$ which will
then be split equally among the first arrivers. The type of problem, known as
\textit{competitive search,} was first posed by Nakai (1986) and has most
recently been updated by Duvocelle et al (2021), with other results in
between. This type of search  models patent races, predators searching for
prey, or mathematicians working on one of the Millenium Prize problems.  In a
military setting, the treasure might be a downed pilot, whom each side in the
conflict wishes to find first. A more general setting is that of so called
winner-take-all games (see Alpern and Howard (2017) where the player with the best
score wins (lowest search time, in our setting).

In this paper we assume the Searchers all navigate with what was introduced by
Alpern (2023) as a \textit{faulty satnav (GPS)}. This means that at every
branch node a pointer (the satnav) suggests an incident edge which goes along
the shortest path to $H$ - but is only correct with a given
\textit{reliability probability} $p<1.$ If incorrect it points randomly to
another incident edge. The only decision variable for the Searcher is the
\textit{trust probability} $q$ with which to follow the pointer; with
complementary probability $1-q$ it randomly picks one of the other edges.  We
assume that once a pointer direction is chosen as the chance move, it is fixed
throughout the play of the game. For example, $H$ might be the Zoo and some
fraction $1-p$ of traffic arrows for the Zoo might be blown from their position in
a morning storm and not be reset until the next day. So a Searcher who returns
to the same node finds the same pointer.  Note that if the Searcher always
follows the pointer he may never reach the target if the pointer is wrong, so
in general $q=1$ is not optimal even if there is a single Searcher who
simply wishes to reach $H$ in the least expected time. We take the network of
search to be a star with $k+1$ rays (one leads to $H,$ the other $k$ do not).
We seek (and uniquely find) a trust $q$ which forms a symmetric equilibrium -
if $n-1$ Searchers trust the pointer with probability $q,$ it is best for the
remaining Searcher with a potentially different trust probability $r$ to also do so, i.e. to set $r=q$. We shall often denote the unique equilbrium value of the trust probability $q$ by $\bar{q}$. An example with $n=2$ Searchers and $k=1$
was given in Alpern (2023), which could be solved with two person zero-sum
methods, though most of that paper was devoted to the case of a single
Searcher who wished to simply minimize the expected time to reach $H.$ We assume that all Searchers have the same brand of Satnav, so they all receive the same pointer.

Our main result, Theorem 1, says that for fixed parameters $n,k$ and $p,$ a
trust probability $q$ forms a symmetric equilbrium if and only if is satisfies
an equation $E_{n,k,p}(q)=0$ given explicitly in equation (\ref{E=0}). Futhermore there is
exactly one solution, so there is a unique symmetric equilibrium.

Our first result on competitive search is that for large $n$ the equilibrium
trust probability $q$ is the same as the reliability $p.$ This result is a
rational case of what is known as \textit{probability matching}. In general
this is an irrational behavior which chooses between two alternatives with the
probabilities with which each is the best (reward 1 rather than 0) - the
rational choice is to always choose the one with the higher probability of being the best one. There are
many documented examples where humans or animals follow this irrational
behavior. See, for example, Shanks et al. (2002), Saldand et al. (2022) and Arieli et al. (2022). We then consider the more difficult problem of an arbitrary number $n$
of Searchers and find that for any parameters $k$ and $p$ there is a unique
symmetric equilibrium where every Searcher trusts with the same probability
$q,$ which is increasing in $p.$ We then show that as the number $n$ of
Searchers increases the equilibrium $q$ eventually decreases to the already found
(probability matching) limit of $p.$ The intuition for this is that a large
number of Searchers will be following the pointer so even if it leads to the
target $H,$ there will be many to share the prize with. 

\section{Literature Review}

We separate our review of previous work into two categories: competitive
search and probability matching. 

\subsection{Literature on Competitive Search}

We believe the first paper to consider competitive search was Nakai (1986),
where two Searchers compete to be first to find a target moving among
discrete locations. A subsequent paper, Nakai (1991), allows the Searchers to
have distinct targets. Garnaev (2007) solved a two person competitive game on
a finite set of locations where players choose effort allocations to be first
to find the treasure. Flesch et al (2009) let the target move between two
locations according to a Markov process. Alpern and Zeng (2021) considered a
variation with two Searchers and a moving hider, where each Searcher wants to
be the first to find the hider. A variant is that of Angelopolous and
Lidbetter (2020) where a Searcher can be thought of competing with an ideal
Searcher who knows the location of the target (aiming to minimize the so called
competitive ratio). Finally, Duvocelle et al (2021) involved two Searchers
competing to first find a target which moves on a time varying Markov process
with a finite number of states. We note that competitive search is a special
case of the winner-take-all game proposed by Alpern and Howard (2017) where
players compete to obtain the lowest score (search time) by choosing among
their available score distributions. (As an example, it two players are
allowed to pick any weighted die with sides labelled 1 to 6 and the usual mean
7/2, they should pick the standard uniform die.)

Competitive search belongs to the more general concept of a \textit{search
game}. See for example the papers Zoroa et al (2011), Clarkson et al (2023)
and the monographs and surverys of Gal (1980), Alpern and Gal (2006), Garnaev
(2012) and Hohzaki (2016). The problem faced by an individual Searcher in competitive search can be phrased as maximizing the probability of arriving before the first of the others arrive (or at the same time). So it is a problem of finding an object by a deadline. However the deadline is unknown. For a related search problem see Lin and Singham (2016).

A recent example of competitive search, mentioned in the Introduction, is the game resulting when a pilot's plane is downed over enemy territory. The resulting race to find him first is a competitive game between the pilot's country and the enemy country. See for example Mouton et al (2015). A historical study if what has been named Combat Search and Rescue (CSAR) is given in Galdorisi and Phillips (2009). These analyses appear to involve a single decision maker rather than our game theoretic analysis, which we believe is novel.

Our paper, in particular, involves searching with hints, or predictions (the pointer). This is a recent innovation in search theory, pioneered by Alpern and Lidbetter (2023) and Angelopolour et al (2026). Both of these papers give a prediction of the Hider's location to the Searcher. The particular model of hints is that introduced by Alpern (2023) as the Fault Satnav (GPS) model.

\subsection{Literature on Probability Matching}\label{seclit2.2}

As stated above, one of our results (Proposition 1) says that when there are
many Searchers ($n\rightarrow\infty)$ the equilibrium trust $q$ is equal to
the signal reliability $p.$ This is an example where probability matching,
which is usually irrational, represents an equilibrium. A useful definition of
probability matching is given in Saldana et al (2022): "In a typical
probability learning task participants are presented with a repeated choice
between two response alternatives, one of which has a higher payoff
probability than the other. Rational choice theory requires that participants
should eventually allocate all their responses to the high-payoff alternative,
but previous research has found that people fail to maximize their payoffs.
Instead, it is commonly observed that people match their response
probabilities to the payoff probabilities." In a more positive result for
large populations, Baddeley et al (2019) find that , "...in the case of ants,
the fraction of the colony foraging at a given location should be proportional
to the probability that resources will be found there." 
There are very similar results more generally related to foraging at food patches. The Ideal Free Distribution (Fretwell and Lucas, 1969; Cressman et al, 2004; Cressman and K\v{r}ivan, 2006) indicates the densities at the food patches for which all individuals are obtaining the same food intake per unit time. A special case of this general theory is Parker's matching principle (Parker, 1978); when food intake is proportional to available resource, we precisely obtain probability matching.
Arieli et al (2022) consider probability matching (indeed overmatching) in a network context. 

\section{Competitive Search}

In \textit{competitive search}, the first players to reach the unit value
treasure (home node $H)$ split it equally between themselves. This is thus a
constant sum game. We retain the dynamics of the Fauty Satnave (GPS) problem
described in the first paragraph of the paper.

For this section we assume the network is a star with $k+1$ rays, one
leads to node $H$ containing the treasure. Due to the symmetry of the leaf
nodes for the star, we could equally interpret the problem as starting with
Nature choosing a random leaf node equiprobably to put the treasure at and
call it $H.$ Since there are $k+1$ rays and one leads to $H,$ a useful signal
can be assumed to satisfy the inequality
\begin{equation}
p>1/\left(  k+1\right)  ,\label{p greater}%
\end{equation}
an assumption we make throughout this section. (In a one shot problem, aiming
simply to reach $H,$ the pointer should be followed if this assumption holds).
We note that a two person zero sum example for the three node line graph
($k=1$ in our notation) was given in Alpern (2023).

The pointer given by the faulty satnav suggests the correct ray with
probability $p$ and every other ray with probability $\left(  1-p\right)  /k,$
which is the reason we adopt the simplifying notations%

\[
p^{\ast}=\frac{1-p}{k},~q^{\ast}=\frac{1-q}{k},~r^{\ast}=\frac{1-r}{k}.
\]

\textbf{Definition}: Given fixed parameters $n,$ $k$ and $p,$ a trust
probability $q=\bar{q}$ is called a symmetric equilibrium if when $n-1$ of the
players trust with probability $q,$ the expected payoff $R$ of the remaining
(focal) player is maximized when also trusting with probability $q.$ Clearly
at equilibrium the symmetry gives the expected payoff for every player as $1/n.$

\subsection{Probability Matching is an Equilibrium for Large $n.$}

It is interesting to note that for large $n,$ the symmetric equilibrium is the
same as the reliability probability $p.$ This gives an example of what is
called probability matching, as discussed in Section \ref{seclit2.2}
It turns out, not surprisingly, that taking $n$
large makes the problem easier to solve because in this case we can assume
that the treasure node $H$ is reached by at least one Searcher on the first
move. Later we will solve the game for all $n$, and then in Theorem 3 we will
show that the equilibrium approaches $p$ as $n$ goes to infinity. That formal
proof does not require the additional assumption that some Searcher
reaches $H$ on the first move. \\

\textbf{Proposition 1:} For large $n$ any $k$ and $p\geq1/\left(  k+1\right)
$ the unique symmetric equilibrium is given by $\bar{q}=p.$

\begin{proof}
For any $q$ such that $0<q<1$ and sufficiently large $n$ we can assume that at least one Searcher successfuly finds $H$ after a single step, so single step probabilities are all that we need to consider.

We consider a focal individual choosing trust probability $r$, when all others choose $q$. Our focal individual then receives a payoff of:
$$
R_{N,k,p,q}(r) = 0 \times \left( 1-p r -(1-p)r^{*} \right) + p r \sum _{m=0}^{n-1} \frac{1}{m+1} q^{m} (1-q)^{n-m-1} {n-1 \choose m}+
$$
$$
 (1-p) r^{*} \sum _{m=0}^{n-1} \frac{1}{m+1} (q^{*})^{m} (1-q^{*})^{n-m-1} {n-1 \choose m}=
$$
\begin{equation}
 \frac{pr}{q} \sum _{m=0}^{n-1} q^{m+1} (1-q)^{n-m-1} \frac{1}{n} {n \choose m+1}+  \frac{(1-p) r^{*} }{q^{*}}\sum _{m=0}^{n-1} (q^{*})^{m+1} (1-q^{*})^{n-m-1} \frac{1}{n} {n \choose m+1}.
\end{equation}
For n sufficiently large and $q$ sufficiently non-extreme to make the chance of all $n-1$ others going to the same site negligible, this implies that
\begin{equation}
R_{n,k,p,q}(r) \approx \frac{1}{n} \left( \frac{pr}{q}+\frac{(1-p)r^{*}}{q^{*}} \right)=\frac{1}{n} \left( \frac{pr}{q}+\frac{(1-p)(1-r)}{(1-q)} \right) .
\end{equation}
Differentiating with respect to $r$ gives a maximum payoff at $r=0$ if $q>p$, at $r=1$ if $q<p$ and any $r$ at $q=p$. Thus $\bar{q}=p$ is the unique Nash equilibrium.
\end{proof}

It is easy to see also that this is actually an Evolutionarily stable strategy (ESS) as the calculations mirror those of the classical sex-ratio problem (see Chapter 4 of  Broom and Rycht\'{a}\v{r}, 2022). Proposition 1 says that probability matching is a rational behavior when $n$ is large.

\subsection{Existence and Uniqueness of the Symmetric Equilibrium}

In this section we consider an arbitrary (possibly small) number $n$ of
Searchers. We derive the equilibrium equation $E\left(  q\right)  =0,$ for
fixed parameters $n,k,p,$ where $E_{n,k,p}\left(  q\right)$ is given by
\begin{equation}
E_{n,k,p}(q)=pq^{\ast}\left(  1-\left(  1-q^{\ast}\right)  ^{n}\right)  \left(  1-\left(
1-q\right)  ^{n-1}\right)  -p^{\ast}q\left(  1-\left(  1-q\right)
^{n}\right)  \left(  1-\left(  1-q^{\ast}\right)  ^{n-1}\right)  .\label{E=0}%
\end{equation}
We show in the first part of Theorem 1 that for any parameters $n,k,p$ there
is a unique $q=\bar{q}$ solving $E_{n,k,p}(q)=0$ and this is the unique symmetric equilibrium.

We illustrate this for parameters $n=5$ and $k=3$ in Figure 1. We plot the
curves $E_{5,3,p}(q)$ for $p=1/2,$(red,
solid), $p=2/3$ (green, dashed) and $p=3/4$ (blue, dotted) for $1/\left(
k+1\right)  =1/3<q<.8.$ Note that for $q=1$ we have $q^{\ast}=0$ and hence
also $E_{5,3,p}(q)=0.$ The $\bar{q}$ values of these intersections (the symmetric equilibrium
trust) are found by solving $E_{5,3,p}(q)=0$ for $q$ for $p=1/2,2/3,3/4,$ giving respectively $\bar{q}=0.53,$
$0.70,$ $0.78.$

It appears that the $E$ curve is higher for larger values of $p,$ resulting in
intersections with $E_{n,k,p}(q)=0$ further to the right. This will be formally
established later in Theorem 3.

\begin{figure}\label{fig:figure1}
\begin{center}
\includegraphics[width=0.6\textwidth]{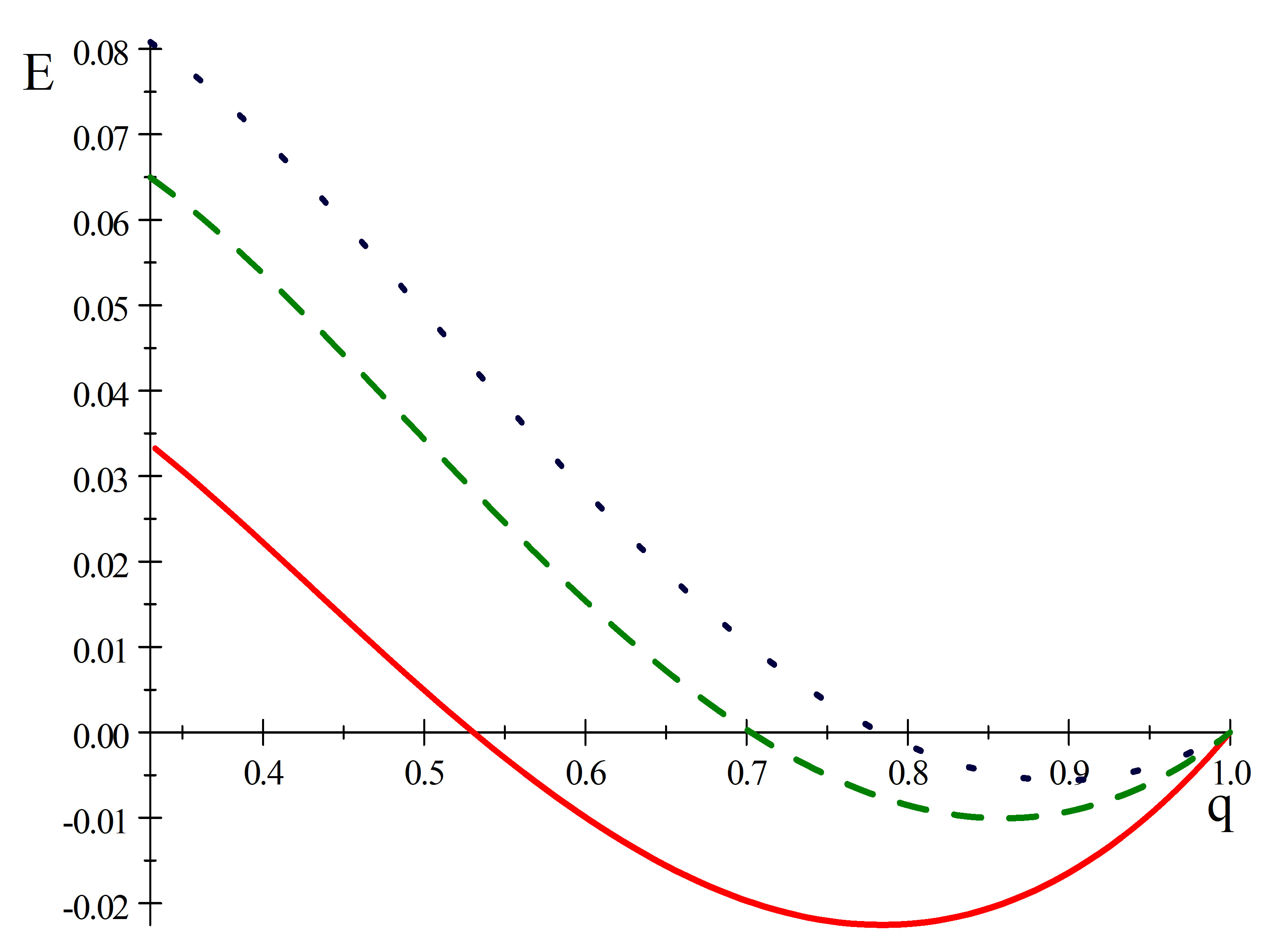}
\end{center}
\caption{Plots of $E\left(  q\right)  $ for $n=5,k=3$ and $p=1/2,~2/3,~3/4.$}
\end{figure}

In Figure 2, we plot the relationship $p=F_{N,k}(q)$ for
parameters $n=5$ and $k=3,$ where $q$ is the symmetric equilibrium trust
probability corresponding to reliability $p.$ To enable comparison with Figure
1, we draw horizontal lines at the three values of $p.$ To illustrate our
later Theorem 2, that equilibrium trust $q$ exceeds reliability $p,$ we also
draw the diagonal line $p=q.$%

\begin{figure}\label{fig:figure2}
\begin{center}
\includegraphics[width=0.6\textwidth]{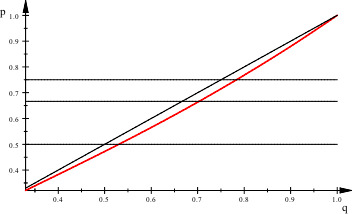}
\end{center}
\caption{Curve $p=F\left(  q\right)  $ (red), $p=q$ (black), $n=5,k=3.$}
\end{figure}

More formally, we have the following.

\begin{theorem}
Let the parameters $n,$ $k$ and $p>1/\left(  k+1\right)  $ be fixed. Then the
trust probability $q$ forms a symmetric equilibrium if and only if the
equilibrium equation $E_{n,k,p}(q)=0$ is
satisfied, where $E(q)$ is given by the formula (\ref{E=0}). \newline
Furthermore, solving $E_{n,k,p}(q)=0$ for $p$ gives a continuous function $p=F_{n,k}(\bar{q})$ (satisfying $E_{n,k,F_{n,k}(\bar{q})}(\bar{q})=0)$ which
is increasing in $\bar{q}.$ Specifically we have%
\begin{equation}
F_{n,k}(q)  =\frac{\left(  1-\left(  1-q^{\ast}\right)  ^{n-1}\right)
\left(  1-\left(  1-q\right)  ^{n}\right)  q}{\left(  1-\left(  1-q^{\ast
}\right)  ^{n-1}\right)  \left(  1-\left(  1-q\right)  ^{n}\right)  q+\left(
1-\left(  1-q^{\ast}\right)  ^{n}\right)  \left(  1-\left(  1-q\right)
^{n-1}\right)  \left(  1-q\right)  }.\label{F}%
\end{equation}
In particular, for any $n,k,p$ there exists a unique equilibrium trust $\bar{q}$, which is increasing in $p$.
\end{theorem}

\begin{proof}
If the pointer is correct, the expected reward that the focal player gets from the next turn is
$$
\sum_{m=0}^{n-1} \frac{1}{m+1} {n-1 \choose m} q^{m}(1-q)^{n-m-1}r=\frac{r}{nq} \sum_{m=0}^{n-1} {n \choose m+1} q^{m+1}(1-q)^{n-(m+1)}=
$$ 
\begin{equation}
\frac{r}{nq} \sum_{j=1}^{n} {n \choose m} q^{m}(1-q)^{n-m}=\frac{r(1-(1-q)^{n})}{nq}.
\end{equation}
This occurs with probability $p$. Analogously, we obtain the expected reward if the pointer is incorrect, occuring with probability $1-p$, of 
$$
\frac{r^{*}(1-(1-q^{*})^{n})}{nq^{*}} .
$$ 
The probability that nobody finds the reward in a given turn is $(1-q)^{n-1}(1-r)$ if the pointer is correct and is $(1-q^{*})^{n-1}(1-r^{*})$ if the pointer is incorrect.

Taken together, we get the expected reward from the game for our focal player as
\begin{equation}
R_{n,k,p,q}(r)=\frac{\frac{pr(1-(1-q)^{n})}{nq}}{1-(1-q)^{n-1}(1-r)}+\frac{\frac{((1-p)r^{*}(1-(1-q^{*})^{n})}{nq^{*}}}{1-(1-q^{*})^{n-1}(1-r^{*})}.
\end{equation}

Differentiating with respect to $r$ gives
$$
\frac{p}{q} \frac{[1-(1-q)^{n}][1-(1-q)^{n-1}(1-r)]-(1-q)^{n-1}r[1-(1-q)^{n}]}{[1-(1-q)^{n-1}(1-r)]^{2}}+
$$
\begin{equation}
\frac{1-p}{q^{*}} \frac{[1-(1-q^{*})^{n}][1-(1-q^{*})^{n-1}(1-r^{*})]-(1-q^{*})^{n-1}r^{*}[1-(1-q^{*})^{n}]}{[1-(1-q^{*})^{n-1}(1-r^{*})]^{2}} \frac{d r^{*}}{dr}=0.
\end{equation}
Substituting $r=q$ and simplifying, noting that $d r^{*}/dr=-1/k$, gives us
\begin{equation}\label{eq:Themainequality}
R_{n,k,p,q}(r)=pq^{*}(1-(1-q^{*})^{n})(1-(1-q)^{n-1})-p^{*}q(1-(1-q)^{n})(1-(1-q^{*})^{n-1})=0,
\end{equation}
as required.

We now consider the theorem's second assertion.
The solution of equation (\ref{eq:Themainequality}) satisfies
\begin{equation}
pq^{*}(1-(1-q^{*})^{n})(1-(1-q)^{n-1})=p^{*}q(1-(1-q)^{n})(1-(1-q^{*})^{n-1}).
\end{equation}
Recalling equation (\ref{F}) and denoting $f(a,b)=(1-(1-a)^{n})(1-(1-b)^{n-1})$, the above implies
\begin{equation}\label{eq:pequation}
p=F_{N,k}(q)=
\frac{qf(q,q^{*})}{(1-q)f(q^{*},q)+qf(q,q^{*})}=
\frac{1}{1+\frac{(1-q) f(q^{*},q)}{q f(q,q^{*})}}.
\end{equation}
To show that $p$ is increasing with $q$ we need to show that $(1-q) f(q^{*},q)/qf(q,q^{*})$ decreases with $q$. We have that
\begin{equation}\label{eq:fratio}
\frac{(1-q)f(q^{*},q)}{qf(q,q^{*})}=\frac{q^{*}}{1-(1-q^{*})^{N-1}}\frac{1-(1-q)^{N-1}}{q}\frac{1-(1-q^{*})^{N}}{1-(1-q)^{N}},
\end{equation}
since from the standard geometric series summation $\sum_{i=0}^{k-1} r^{i}=(1-r^{k})/(1-r)$, we have that
\begin{equation}\label{eq:fratio2}
\frac{f(q^{*},q)}{f(q,q^{*})}=\frac{\sum_{i=0}^{n-2} (1-q)^{i}}{\sum_{i=0}^{n-2} (1-q^{*})^{i}}  \frac{1-(1-q^{*})^{n}}{1-(1-q)^{n}}.
\end{equation}
We know that $1-q$ is decreasing in $q$ and $1-q^{*}$ is increasing in $q$. The first ratio above is then clearly a summation decreasing in $q$ divided by one increasing in $q$. The second ratio is also clearly a term decreasing in $q$ divided by one increasing in $q$. As all terms are positive, the whole expression is decreasing in $q$, hence $p$ is increasing in $q$.

As $q \rightarrow 1$, $q^{*} \rightarrow 0$, so that the ratio in equation ( \ref{eq:fratio2}) tends to 0, we have from equation (\ref{eq:pequation}) that $p \rightarrow 1$. Similarly, as $q \rightarrow 1/(k+1)$, $q^{*} \rightarrow 1/(k+1)$, so that the ratio in equation (\ref{eq:fratio2}) tends to 1, yielding from equation (\ref{eq:pequation}) that $p \rightarrow 1/(k+1)$.

Thus there is a 1-1 map between $p$ and $q$ in the range $[1/(k+1), 1]$ and so a unique $q$ as a function of $p$.

\end{proof}

\subsection{At Equilibrium, Signal Reliability Exceeds Trust}

We showed in Proposition 1 that for large (infinite) $n,$ equilibrium trust is
equal to the signal reliability probability $p,$ an example of what is known
as probability matching. In this section we show that for finite $n$, at
equilibrium the Searchers should trust a signal of reliability $p$ with a
probability that \textit{exceeds} $p.$ This is already apparent for particular
parameters from Figure 2, where the red curve $p=F\left(  q\right)  $ lies
\textit{below} the line $p=q,$ which it intersects at $q=1.$ We shall see in
the next section that the degree by which $p$ lies below the equilibrium value $\bar{q}$ eventually decreases as the
number $n$ of Searchers increases (we say a sequence $a(n)$ is eventually decreasing if  for some $N$, $a(n+1) \leq a(n)$ for all $n \geq N$). 

\begin{theorem}
At equilibrium, trust $\bar{q}$ exceeds signal reliability $p.$ That is, for any
parameters $n$ and $k,$ $\bar{q}>p=F_{n,k}(\bar{q}).$
\end{theorem}

\begin{proof}
Recall our formula $F$ for $p$ in terms of $q$ given in (\ref{F}) and (\ref{eq:pequation}), which implies
that%
\begin{equation}
\frac{p}{q}=\frac{f(q,q^{*})}{qf(q,q^{*})+\left(  1-q\right)  f(q^{*},q)}<1,\text{ if }f(q,q^{*})<f(q^{*},q).
\end{equation}
We now show that $f(q,q^{*})<f(q^{*},q)$ which will imply that $p/q<1,$ or $p<q.$ We set
$x\equiv1-q$ and $y\equiv1-q^{\ast}.$ Since $p$ and $q$ are assumed to be at
least $1/\left(  k+1\right)  ,$ it follows that $q^{*}\leq q$ and so $x \leq y.$ In
this notation we have%
\[
f(q,q^{*})-f(q^{*},q)=g\left(  x,y\right)  \equiv\left(  1-x^{n-1}\right)  \left(
1-y^{n}\right)  -\left(  1-x^{n}\right)  \left(  1-y^{n-1}\right).
\]
It remains only to show that $g\left(  x,y\right)  >0$ for $0<x<y<1.$ First
note that for any $x$ we have $g\left(  x,y\right)  =0$ for $y=x$ and for
$y=1.$ We will show that for any fixed $x$, $g\left(  x,y\right)  $ is
positive along the line $y$ from $x$ to $1.$ Setting $h\left(  y\right)
=g\left(  x,y\right)  $ we see that $h$ has zeros at $y=x$ and $y=1$ so by
Rolle's Theorem it has a critical point $\bar{y}$ in $\left(  x,1\right)  $,
which we can determine by calculating%
\begin{align*}
h^{\prime}\left(  y\right)    & =\left(  n-1\right)  ~\left(  1-x^{n}\right)
~y^{n-2}-n\left(  1-x^{n-1}\right)  y^{n-1},\text{ with}\\
h^{\prime}\left(  y\right)    & =0\text{ uniquely for }\bar{y}=\frac{n-1}%
{n}\frac{\left(  1-x^{n}\right)  }{\left(  1-x^{n-1}\right)  }<1.
\end{align*}
The uniqueness of $\bar{y}$ ensures there cannot be any zeros of $h$ between
$x$ and $1.$ To determine whether $h$ (and hence $g)$ is negative or positive
on $\left[  x,1\right]  ,$ we calculate $h^{\prime}\left(  1\right)  =\left(
n-1\right)  ~\left(  1-x^{n}\right)  -n\left(  1-x^{n-1}\right)  <0$ for
$x<1.$ This follows from our earlier observation that $\bar{y}<1.$ This
implies $h$ is positive for $x$ close to 1. So $h$ and hence $g$ is positive
and so $f(q,q^{*})<f(q^{*},q)$ and hence $p<\bar{q}.$
\end{proof}

\subsection{Equilibrium Trust is eventually decreasing in the Number of Searchers}

We now consider the effect of population size $n$ on the equilibrium trust. In
Figure 3 we plot $p=F_{n,k}(q)$ for $k=3$ and various values on
$n.$ Since $p$ is a given parameter (independent variable), we also plot for
convenience in Figure 4 the same curve with the dependent variable $q=\bar{q}$
on the vertical axis. From Figure 4 we see that the equilibrium $\bar{q}$ is decreasing
in $n$ for all $p.$ In fact this pattern may not hold for small values of $n.$%

\begin{figure}\label{fig:figure3}
\begin{center}
\includegraphics[width=0.6\textwidth]{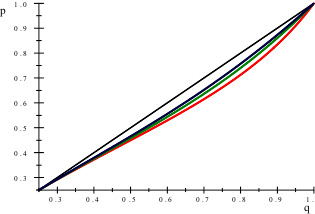}
\end{center}
\caption{Plot of $p=F\left(  q\right)  ,n=2$(red), $3$ (green), $4$ (blue)
for $k=3$.}
\end{figure}

\begin{figure}\label{fig:figure4}
\begin{center}
\includegraphics[width=0.6\textwidth]{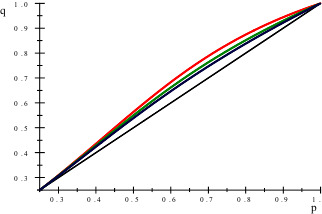}
\end{center}
\caption{Plot of $\bar{q}=F^{-1}\left(  p\right)  ,k=3$ and $n=2$(red), $3$
(green), $4$ (blue).}
\end{figure}

\begin{theorem}
~For fixed $k,p,$ the equilibrium trust $\bar{q}$ is eventually decreasing in $n$
and converges to $p.$ The last part is a formalization of Proposition 1.
\end{theorem}

\begin{proof}
Now consider $E_{n, p, k}(q)-E_{n+1, p, k}(q)$. From equation (\ref{eq:Themainequality}) we have
$$
E_{n,k,p}(q)-E_{n+1,k,p}(q)=pq^{*}[(1-(1-q^{*})^{n})(1-(1-q)^{n-1})-(1-(1-q^{*})^{n+1})(1-(1-q)^{n})]+
$$
\begin{equation}\label{eq:Ediffeq}
p^{*}q[(1-(1-q)^{n+1})(1-(1-q^{*})^{n})-(1-(1-q)^{n})(1-(1-q^{*})^{n-1})].
\end{equation}
From Theorem 2 we have that $p<q$ so that $pq^{*}<p^{*}q$. Together with the fact that 
$(1-(1-q)^{n+1})(1-(1-q^{*})^{n})-(1-(1-q)^{n})(1-(1-q^{*})^{n-1})>0$, replacing $p^{*}q$ by $pq^{*}$ in equation (\ref{eq:Ediffeq}) makes the expression smaller and we obtain
$$
\frac{E_{n,k,p}(q)-E_{n+1,k,p}(q)}{pq^{*}}>[(1-(1-q^{*})^{n})(1-(1-q)^{n-1})-(1-(1-q^{*})^{n+1})(1-(1-q)^{n})]+
$$
\begin{equation}\label{eq:T1s4ineq}
[(1-(1-q)^{n+1})(1-(1-q^{*})^{n})-(1-(1-q)^{n})(1-(1-q^{*})^{n-1})].
\end{equation}
We now proceed to manipulate the right hand side of inequality (\ref{eq:T1s4ineq}). It is easy to see that it rearranges to 
$$
(1-(1-q^{*})^{n})[2-(1-q)^{n-1}-(1-q)^{n+1}]-(1-(1-q)^{n})[2-(1-q^{*})^{n-1})-(1-q^{*})^{n+1}].
$$
Adding and subtracting the term $2(1-(1-q^{*})^{n})(1-(1-q)^{n})$ allows rearrangement of the above to give
$$
-(1-(1-q^{*})^{n})(1-q)^{n-1}(1-(1-q))^{2}+ (1-(1-q)^{n})(1-q^{*})^{n-1})(1-(1-q^{*}))^{2}=
$$
$$
-(1-(1-q^{*})^{n})(1-q)^{n-1}q^{2}+(1-(1-q)^{n})(1-q^{*})^{n-1}(q^{*})^{2}.
$$
This then rearranges to 
$$
q^{2}(q^{*})^{2}[-(1-(1-q^{*})^{n})(1-q)^{n-1} \left( \frac{k}{1-q} \right)^{2}+(1-(1-q)^{n})(1-q^{*})^{n-1}\frac{1}{q^{2}}].
$$
Since $q+q^{*}<1$ we have that $1/q > 1/(1-q^{*})$. We also know that $q>p>p^{*}>q^{*}$ whenever $p>1/(k+1)$ from Theorem 2. Using there two statements, we have that the above expression is greater than
$$
q^{2}(q^{*})^{2}((1-(1-q)^{n})[(1-q^{*})^{n-3}-(1-q)^{n-3}k^{2}]>q^{2}(q^{*})^{2}((1-(1-q)^{n})[(1-p^{*})^{n-3}-(1-p)^{n-3}k^{2}].
$$
The above chain of reasoning thus implies that
\begin{equation}\label{eq:newEinequality}
\frac{E_{n,k,p}(q)-E_{n+1,k,p}(q)}{pq^{*}}>q^{2}(q^{*})^{2}((1-(1-q)^{n})[(1-p^{*})^{n-3}-(1-p)^{n-3}k^{2}].
\end{equation}
The right hand side of equation (\ref{eq:newEinequality}) is positive if the term in the square brackets is positive. This is true whenever
\begin{equation}\label{eq:Npkcondition}
n>n^{*}(p,k)=3+\frac{2 \ln k}{\ln \frac{k-1+p}{k(1-p)}}.
\end{equation}
We note that $n^{*}(p,k)$ is a finite positive number, whenever $p>1/(k+1)$. If $k=1$ or $p \rightarrow 1$ then this value is 3. As $p \rightarrow 1/(k+1)$ it tends to $\infty$.

We thus have that 
\begin{equation}
E_{n,k,p}(q)-E_{n+1,k,p}(q)>0,
\end{equation}
whenever the condition from inequality (\ref{eq:Npkcondition}) holds.

Now suppose that for a particular $n, p, k$ we have that $E_{n,k,p}(q)=0$ when $q=q(n)$. Clearly, then $E_{n+1,k,p}(q(n))<0$.
We know that $E_{n+1,k,p}(q)$ is an increasing function of $p$, so that the $p'$ for which $E_{n+1,k,p'}(q(n))=0$ is larger than $p$. From Theorem 1 we know that $p$ is increasing in $q$, so that the value $q(n+1)$ for which $E(n+1, p, k, q(n+1))=0$ must satisfy 
\begin{equation}
q(n+1)<q(n),
\end{equation}
whenever inequality (\ref{eq:Npkcondition}) holds, which is true for sufficiently large $n$.

Thus, as $n \rightarrow \infty$, $q(n)$ is a decreasing sequence which is bounded below, by $p$. It therefore converges to some number no smaller than $p$. Letting $n \rightarrow \infty$ in equation (\ref{eq:Themainequality}) means that $q$ gets arbitarily close to $p$, since the limiting equation is 
\begin{equation}
E_{n,k,p}(q)=pq^{*}-p^{*}q=0.
\end{equation}
Thus the limiting value of $q$ must be $p$.
\end{proof}

\subsection{Equilibrium Trust is Increasing in the Number of Rays}

As the number of rays in the star increases, the value of the signal becomes
more important as it is selecting one of many possiblities. So for example
knowing that a particular ray leads to the treasure with probability one half
when there are hundreds of rays is very useful. We formally establish this
intuition in the following result. 

\begin{theorem}
For fixed $n,p,$ the equilibrium trust $\bar{q}$ is increasing in $k.$
\end{theorem}

\begin{proof}
Using equation (\ref{eq:Themainequality}) from statement 1, $E(q)=0$ can be re-organised as follows:
\begin{equation}\label{eq:Themainequalitybs4}
\frac{p}{1-p}=\frac{q}{1-q} \frac{1-(1-q)^{n}}{1-(1-q)^{n-1}} \frac{1-(1-q^{*})^{n-1}}{1-(1-q^{*})^{n}}.
\end{equation}

The left-hand side of equation (\ref{eq:Themainequalitybs4}) is increasing in $p$ and so this equation gives us $p$ as a function of $q$ and $k$ (through $q^{*}$), which we shall denote $p(q,k)$.

Differentiating the right-hand side of (\ref{eq:Themainequalitybs4}) with respect to $k$ gives
$$
\frac{d}{dk} \left( \frac{p}{1-p} \right) = \left( \frac{1}{1-(1-q^{*})^{n}} \right)^{2}\frac{1-q}{k^{2}} \times 
$$
$$
\left(-(1-(1-q^{*})^{n}) (n-1)(1-q^{*})^{N-2}+ (1-(1-q^{*})^{n-1}) n(1-q^{*})^{n-1} \right) = 
$$
\begin{equation}
\frac{(1-q^{*})^{n-2}}{(1-(1-q^{*})^{n})^{2}}\frac{1-q}{k^{2}} \left(-(n-1)-(1-q^{*})^{n}+n(1-q^{*}). \right) 
\end{equation}

Since $nx-(n-1)-x^{n}=-(1-x)(n-1-x-x^{2}-\ldots -x^{n-1})<0$ for $0 \leq x<1$, we have that $p(q,k)$ is decreasing in $k$. In particular, $p(q,k+1)< p(q,k)$. 

For fixed $k$ we know that $p$ increases with $q$, so to find $q'$ such that $p(q',k+1)=p(q,k)$, we need $q'>q$. Thus $q$ is increasing with $k$.
\end{proof}

It is interesting to compare this result with that of Alpern (2023) for a
single Searcher trying to minimize the expected time to reach target $H,$
where the optimal trust is given by $q\left(  p,k\right)  =\left(  p-\sqrt
{k}\sqrt{p\left(  1-p\right)  }\right)  /\left(  1-(k+1)(1-p)\right)  $ which
is decreasing in rays $k.$ 

\section{Conclusions}

Network search games have until now considered conflict between a Searcher and an immobile or mobile Hider, over the time (search time) T until the Hider is found. This is the first paper to consider a competitive search on a network among Searchers to be the first to find the Hider, although this has been studied a little in a non-network setting, see Nakai (1989) and Duvocell et al (2021). Our main result is that for symmetric Searchers (starting at the same point, using the same Satnav) there is always a unique symmetric equilibrium, where all Searchers trust the given directions with the same trust probability. Further, this equilibrium level increases with the probability that the directions are correct.

An interesting finding in the case of a large, or infinite, population of Searchers is that the equilibrium is to follow the suggested branch with the probability that it is correct. This is an example of what is known as probability matching. Usually this strategy is irrational (sub optimal). For smaller populations the trust level exceeds this probability, but (eventually) decreases with the population size, converging to it as the population becomes large. The trust probability also increases with the number of alternative directions that could be chosen.

We should also comment on the generality of our network, the star network. While very specific, this network was used in the related optimization scenario of the Faulty Satnav Problem of Alpern (2023) to begin a solution for more general networks. After the star was solved there, that solution was extended to networks with a bridge edge to the destination node  H. Then this was extended to tree networks, where every edge is a bridge. It is well known that game problems are harder to solve than optimization ones and so this paper is restricted to the case of a star, with similar extensions hoped for in the future.

Another extension is to allow the trust probability to depend on the time the search has been going on for. If a Searcher knows the time (how long they have been searching) they might decide to trust less over time, as it becomes more likely that the Satnav (GPS) has been wrong. \\

\hspace{-0.7cm} {\bf Acknowledgements:} Mark Broom was supported by the European Union’s Horizon 2020 research and innovation programme under the Marie Sklodowska-Curie grant agreement No 955708.

\end{document}